\documentclass{article}

\usepackage[dvips]{graphicx}
\usepackage[hyperindex,breaklinks]{hyperref}
\usepackage{apacite}
\usepackage{amsfonts}
\usepackage{units}
\usepackage{amsmath,amssymb,color}
\usepackage{graphicx}
\usepackage{booktabs}
\usepackage{multirow}
\usepackage[nodisplayskipstretch]{setspace}
\usepackage[capposition=top]{floatrow}
\usepackage{amsthm}

\newcommand{\bfone}{\mbox{$\mathbf 1$}}
\newcommand{\bfzero}{\mbox{$\mathbf 0$}}

\newcommand{\cc}{\mbox{$\mathbf c$}}

\def\gg{\mbox{$\mathbf g$}}

\def\qq{\mbox{$\mathbf q$}}
\def\sb{\mbox{$\mathbf s$}}

\newcommand{\x}{\mbox{$\mathbf x$}}
\def\XX{\mbox{$\mathbf X$}}
\def\xx{\mbox{$\mathbf x$}}

\newcommand{\yy}{\mbox{$\mathbf y$}}

\newcommand{\aalpha}{\boldsymbol{\alpha}}
\newcommand{\nnu}{\boldsymbol{\nu}}
\newcommand{\xxi}{\boldsymbol{\xi}}
\newcommand{\zzeta}{\boldsymbol{\zeta}}

\newtheorem{corollary}{Corollary}
\newtheorem{definition}{Definition}

\newtheorem{proposition}{Proposition}
\newtheorem{theorem}{Theorem}

\begin{document}


\title{Non-identifiability, equivalence classes, and attribute-specific classification in Q-matrix based Cognitive Diagnosis Models}


\author{Stephanie S.\ Zhang, Lawrence T.\ DeCarlo, and Zhiliang Ying\\Columbia University}
\maketitle


\begin{abstract}
There has been growing interest in recent years in Q-matrix based cognitive diagnosis models. Parameter estimation and respondent classification under these models may suffer due to identifiability issues. Non-identifiability can be described by a partition separating attribute profiles into groups of those with identical likelihoods. Marginal identifiability concerns the identifiability of individual attributes. Maximum likelihood estimation of the proportion of respondents within each equivalence class is consistent, making possible a new measure of assessment quality reporting the proportion of respondents for whom each individual attribute is marginally identifiable. Arising from this is a new posterior-based classification method adjusting for non-identifiability.
\end{abstract}

\section{Introduction}
Diagnostic assessments are created with the goal of making classification-based decisions about respondents' possession  of multiple latent traits, also known as attributes. Researchers have brought a number of tools to bear on the problem of diagnostic classification, including multidimensional IRT, factor analysis,  the rule-space method, the attribute hierarchy method, clustering methods, and cognitive diagnosis models (CDMs); for a recent review, see \citeA{Rupp}. CDMs are multidimensional latent variable models with a vector of binary latent variables representing mastery of a finite set of skills whose analysis results in a probabilistic attribute profile; this makes them well-suited to diagnostic classification. Well known models include the Deterministic Input, Noisy ``And'' Gate (DINA) model, the Deterministic Input, Noisy ``Or'' Gate (DINO) model, the Noisy Inputs, Deterministic ``And'' Gate (NIDA) model, the Noisy Inputs, Deterministic ``Or'' Gate (NIDO) model, and the Conjunctive Reparameterized Unified Model (C-RUM), among others \cite{Rupp, Haertel, Junker, dela,Maris,Templin, TatsuokaC, Templin2006, dela2008}.

The DINA model is one of the best known and widely used CDMs. Underlying the model is the assumption that, before slipping and guessing come into play, a respondent must have mastered all necessary (as specified by a loading matrix known as the Q-matrix) attributes required by a particular item in order to answer that item correctly. Thus it is a conjunctive, non-compensatory model, well-suited to educational assessments in areas such as mathematics where correct answers are obtained by correctly employing all of an item's required skills together. The DINA model has been frequently employed in the analysis of  assessments, including the widely analyzed fraction subtraction data set of \citeA{Tatsuoka1990} (see \citeNP{dela2009, dela, delaDoug08, THD06, deCarlo2011, HensonTemplin09}). However, even after many refinements to the methodology there are still some persistent issues. In the fraction subtraction data, for example, respondents  who answer all items incorrectly are often classified as having most of the skills \cite{deCarlo2011}. Classification issues of this type can result model misspecification, but they can also be the unavoidable consequence of the structure of the assessment. Specifically in the DINA model, attributes that appear solely in conjunction with other attributes are problematic \cite{deCarlo2011}. This is due to an issue with attribute identifiability, which has long been known \cite{Tatsuoka1991,DiBello,deCarlo2011}, but tends to be ignored in practice.

This paper gives a formal treatment of  the identifiability issue of the DINA model and related CDMs. Since an assessment with fully identifiable attributes is often unavailable, we include guidelines for classification under non-identifiability and a consistent measure for the extent of non-identifiability. This allows classification error control and assessment evaluation in terms of identifiability.

The paper begins by reviewing some basic concepts, including the DINA model and its variants, in Section~\ref{sec:background}. We introduce the issue of identifiability for Q-matrix based assessments in Section~\ref{sec:prob}. In Section~\ref{sec:part}, we explain the use of equivalence classes and partitions to fully describe the structure of the attribute profile space, in terms of identifiability; an algorithm to generate the partition, given any Q-matrix, is included. Partitioning allows consistent estimation of the proportion of individuals in each group of equivalent attribute profiles, as explained in Section~\ref{sec:est}. These results are extended to individual attributes via {marginal identifiability} in  Section~\ref{sec:single}. In fact, the consistent estimation of the proportion of the population for which each attribute is marginally identifiable leads to a reliable measure of exam quality (with respect to identifiability). We also create a decision rule for respondent classification which controls misclassification probabilities in Section~\ref{sec:class}. Section~\ref{sec:exts} examines the implications of these methods to several variants of the DINA, in addition to another Q-matrix based CDM, the DINO model. Finally, results derived from both simulation and \citeauthor{Tatsuoka1990}'s fractions data set are reported in Section~\ref{sec:results}.

\section{Background}\label{sec:background}
Throughout the paper, we will be using some standard concepts from the study of CDMs. Some specific terminology and notations are listed below.
\begin{description}
\item[Attributes] are the respondent's (unobserved) mastery of certain skills. If we suppose that there are $N$ respondents and $K$ attributes, let the matrix of attributes be $A = (\alpha_{ik})$, where where $\alpha_{ik} \in \{0,1\}$ indicates the presence or absence of the $k$-th attribute in the $i$-th respondent. An \emph{attribute profile} $\aalpha = (\alpha_1,\ldots, \alpha_K)^\top$ is the vector of all attributes; an individual respondent $i$ will have attribute profile $\aalpha^i$ such that $\alpha^i_k = \alpha_{ik}$.

\item[Responses] are the respondent's binary responses to items. Given $N$ respondents and $J$ items, the responses can be written as a $N\times J$ matrix $X = (X_{ij})$, where  $X_{ij} \in \{0,1\}$ is the response of the $i$-th respondent to the $j$-th item. The $i$-th respondent's responses will be denoted by the vector $\XX^i$, where the $j$-th element $X^i_j = X_{ij}$ for all $i,j$.

\item[The Q-matrix] is the link between the items and their attribute requirements. It is a $J\times K$ matrix $Q = (q_{jk})$, where for each $j,k$, $q_{jk} \in \{0,1\}$ indicates whether the $j$-th item requires the $k$-th attribute. From the Q-matrix we can extract the attribute requirements of an item $j$ as the vector $\qq^j$, where the $k$-th element $q^j_k =  q_{jk}$ for all $j,k$. 
\end{description}

\subsection{The DINA Model}

This paper focuses on the DINA model, one of the most widely used CDMs. Under the DINA model, given an attribute profile $\aalpha$ and a Q-matrix $Q$, we can further define the quantity
$$\xi_j(Q,\aalpha)  = \prod_{k=1}^K (\alpha_k)^{q_{jk}} = \textbf{1}(\alpha_k \geq q_{jk}: k = 1, \ldots, K),$$
which indicates whether a respondent with attribute profile $\aalpha$ possesses all the attributes required for item $j$. If we suppose no uncertainty in the response, then a respondent $i$ with attribute profile $\aalpha$ will have responses $X_{ij} = \xi_j(Q,\aalpha)$ for $j=1,\ldots, J$. Thus, the vector $\xxi = (\xi_1,\ldots, \xi_J)^\top$ is known as the \emph{ideal response vector}. 

In the DINA model, uncertainty is incorporated at the item level. With each item $j = 1, \ldots, J$, we associate a slipping parameter $s_j = P(X_j = 0|\xi_j = 1)$ and a guessing parameter $g_j = P(X_j = 1|\xi_j = 0)$. Each $X_j$ is Bernoulli with success probability $(1-s_j)^{\xi_j}g_j^{1-\xi_j}$. Thus, the probability of a particular response vector $\xx$ given the ideal response vector $\xxi$ is
\begin{align}\label{eqDINA}
P(\xx|\xxi,\sb,\gg) &= \prod_{j=1}^J [(1-s_j)^{\xi_j} g_j^{1-\xi_j}]^{x_j}  [1-(1-s_j)^{\xi_j} g_j^{1-\xi_j}]^{1-x_j}\nonumber\\
 &= \prod_{j=1}^J (1-s_j)^{\xi_jx_j} g_j^{(1-\xi_j)x_j}s_j^{\xi_j(1-x_j)}(1- g_j)^{(1-\xi_j)(1-x_j)}
 \end{align}
In addition to $\sb$ and $\gg$, the response distribution also depends on $\nnu = (\nu_{\aalpha})_{\aalpha\in\{0,1\}^K}$, the proportion of individuals possessing each attribute profile. Generally, diagnostic classification is based on the posterior $p(\aalpha|\xx)$, which is calculated using and can be very sensitive to the prior, $\nnu$. When $\sb$, $\gg$, and $\nnu$ are unknown, these parameters must be simultaneously estimated. 

\subsection{Variants of the DINA Model}\label{sec:DINAvar}
Several variants of the DINA can be constructed by restricting $\nnu$ to some lower-dimensional subspace. For example, assuming independence among the attributes so that
$$\nu_{\aalpha} = p(\aalpha) = \prod_{k=1}^K p(\alpha_k)$$
reduces the $2^K-1$-dimensional parameter space to a $K$-dimensional one. This restriction is referred to as the independent DINA (ind-DINA) from hereon. It is convenient to model each $\alpha_k$ with a logistic link, so that
$$p(\alpha_k) = \exp(\alpha_kb_k)/[1+\exp(b_k)],$$
where $b_k$ denotes the attribute's `difficulty.'

Another alternative is the higher-order DINA (HO-DINA) model \cite{dela, THTR2008}. This model assumes that the probability of possessing a skill is dependent on a continuous skill factor $\theta$ following the standard normal distribution, so that
$$\nu_{\aalpha} = p(\aalpha) = \int_\theta p(\aalpha|\theta)p(\theta) d\theta.$$
Each individual attribute is assumed to be conditionally independent given $\theta$, so that
$$p(\aalpha|\theta) = \prod_{k=1}^K p(\alpha_k|\theta).$$
Finally the individual probabilities $p(\alpha_k|\theta)$ can be modeled with a logistic link,
$$p(\alpha_k|\theta) = \exp(\alpha_k(b_k + a_k\theta))/[1+\exp(b_k+a_k\theta)],$$
where $b_k$ denotes the attribute's `difficulty,' and $a_k$ is the attribute discrimination parameter. It is also possible to fit a restricted version of this model, for which all  the $a_k$ must be equal, as in \citeA{dela}. This is referred to as the restricted higher order DINA (RHO-DINA) model \cite{deCarlo2011}.

\section{Identifiability Issues in the DINA}\label{sec:prob}
Diagnostic assessments are meant to provide detailed information about respondents' possession of a variety of traits. Preferably, a well-designed exam will be able to provide information about each trait for every respondent. However, recovering information about the latent variables from a `0' response may be difficult; in comparison to a `1' response, which suggests that a respondent is more likely to possess each attribute associated with that item, a `0' response may indicate the failure to master ony one or several of the required attributes. Consider the following two simple Q-matrices for the DINA model:
\begin{equation}
Q_1 = \begin{pmatrix}1&0\\0&1\end{pmatrix},\ Q_2 = \begin{pmatrix}1&0\\1&1\end{pmatrix}.
\end{equation}
In assessments based on the Q-matrix $Q_1$, a correct response to each item generally indicates a higher probability that the respondent possesses the corresponding attribute, while an incorrect response indicates a lower probability of the same. However, with $Q_2$, an incorrect response to the second item only implies that at least one of the attributes is probably missing. In fact, given that a student does not possess Attribute 1, Item 2 provides no information about his or her mastery of Attribute 2, and so respondents with attribute profiles  $(0,0)$ and $(0,1)$ have statistically identical responses. Thus, the assessment as a whole is incapable of differentiating between the two profiles, and any classification decision between them will solely be a reflection of the prior information.
 
A slightly more complicated situation appears if we add a third attribute to the example above. Consider an assessment following the DINA model with Q-matrix $Q_3$, where
\begin{equation}
Q_3 = \begin{pmatrix}1&0 & 0 \\1&1 & 0 \\ 0 & 1 & 1\end{pmatrix}.
\label{eq:Q3}
\end{equation}
The attribute requirements of the first two items match those of the items corresponding to $Q_2$. Now, however, the proportion of individuals for whom Attribute 2 is not identifiable is smaller. Of those who do not possess Attribute 1, some will possess Attribute 3. Then Attribute 2 is identifiable because of differing response distributions on Item 3. However, response distributions for those with attribute profiles $(0,1,0)$ and $(0,0,0)$ are still indistinguishable. Thus, although the assessment provides no information about Attribute 2 for a smaller part of the population, the issue has not been completely resolved.

\section{Partitioning the Attribute Profile Space}\label{sec:part}
We begin with an intuitive criterion for deciding whether an assessment has the ability to differentiate between two attribute profiles.
\begin{definition}\label{def:sep}
Two attribute profiles are \emph{separable} if they lead to different response distributions.
\end{definition}
The differing response distributions of separable attribute profiles imply that the data will favor one profile or the other; there is some differential effect on the likelihood and thus the posterior. Profiles that are not separable are statistically identical, with equivalent likelihood functions, making any differences in their posteriors simply artifacts of the prior.

Determining whether attributes are separable can be done without the full response distribution;
in fact, only the ideal responses $\xxi(Q,\aalpha)$ are necessary.
\begin{proposition}\label{thmIdeal}
Given a Q-matrix $Q$ and slipping and guessing parameters $\sb$ and $\gg$, two attribute profiles $\aalpha^1$ and $\aalpha^2$ can be separated if and only if they produce ideal response vectors $\xxi^1 = \xxi(Q,\aalpha^1)$ and $\xxi^2 = (Q,\aalpha^2)$ such that for some $j \in \{1,\ldots, J\}$,  $\xxi^1_j \neq \xxi^2_j$ and $1-s_j \neq g_j$.
\end{proposition}
Throughout the rest of this paper we assume that $1-s_j \neq g_j$ for each $j = 1,\ldots, J$, which simplifies Proposition \ref{thmIdeal} into Corollary \ref{corrIdeal}. Should such an item indeed be present, then it has no discriminating power and may be omitted.
\begin{corollary}\label{corrIdeal}
 If every item $j$ has different success probabilities given $\xi_j = 1$ or given $\xi_j = 0$, i.e.\ $1-s_j \neq g_j$ for $j = 1,\ldots, J$, then two attribute profiles can be separated if and only if they produce different ideal response vectors.
 \end{corollary}

Lastly, it is also of interest whether an attribute profile can be separated from all other attribute profiles, and is thus identifiable. This definition of identifiability will be tied to the general statistical concept in Section~\ref{sec:est}.
\begin{definition}
An attribute profile $\aalpha$ is \emph{identifiable} when it can be separated from any other attribute profile $\aalpha'\neq \aalpha$.
\end{definition}

\subsection{Complete Separation of Attribute Profiles}
The first step in understanding the identifiability issue is determining under what circumstances all attribute profiles are identifiable. This depends on the Q-matrix, which is called complete when it leads to full identifiability\cite{Chiu}. Formally, we have the following definition:
\begin{definition}
Under a \emph{complete} Q-matrix, all attribute profiles are identifiable, i.e.\ $\xxi(Q,\aalpha) \neq \xxi(Q,\aalpha')$ iff $\aalpha \neq \aalpha'$. 
\end{definition}
The requirements for completeness have long been known \cite{Tatsuoka1991,DiBello,Chiu}. In essence, the assessment must contain at least one item devoted solely to each attribute. In terms of the Q-matrix, this means that for each $k \in \{1,\ldots, K\}$, there should be at least one row with an entry of `1' solely in the $k$-th position. 
\begin{proposition}\label{prop:complete}
Let $R_Q$ be the set of row vectors of Q-matrix $Q$. Then $Q$ is \emph{complete} iff $\{e_k: k = 1,\ldots, K\} \subset R_Q$, where $e_k$ is a vector such that the $k$-th element is one and all other elements are zero. 
\end{proposition}

\subsection{Partial Separation of Attribute Profiles}
While a complete Q-matrix is necessary for full identifiability, many of the Q-matrices used in practice are unfortunately incomplete. In fact, creating assessments with complete Q-matrices is oftentimes infeasible, and requiring a complete matrix for analysis would make models like the DINA model highly impractical. This makes the partition, a standard mathematical construct, an essential tool in accurately and systematically describe the structure of nonidentifiability in the DINA. Partitions are formed from equivalence relations, which have the following requirements:

\begin{definition}
The relation `$a \sim b$' is an \emph{equivalence relation} if it is 
\begin{itemize}
\item reflexive: $a \sim a$
\item symmetric: $a \sim b$ iff $b \sim a$
\item transitive: If $a\sim b$ and $b \sim c$, then $a \sim c$.
\end{itemize}
\end{definition}

\begin{proposition}\label{prop:equiv}
Let `$\sim$' denote the binary relation `cannot be separated,' where $\aalpha^1 \sim \aalpha^2$ if and only if $\xxi(Q,\aalpha^1) = \xxi(Q,\aalpha^2)$. Then `$\sim$' is an equivalence relation.
\end{proposition}

Putting profiles into groups, known as equivalence classes, based on an equivalence relation results in a partition; in this case, any two attribute profiles in the same equivalence class cannot be separated, while any two in different classes can be. We denote a particular equivalence class by $[\aalpha]$, where $\aalpha$ may be any attribute profile in the class; literally, $[\aalpha]$ can be read as ``the set of attribute profiles equivalent to $\aalpha$.''

The simplest way of determining the partition would be to calculate the ideal response vector of  each of the $2^K$ attribute profiles and sort them lexicographically. This runs quickly in $\mathcal{O}(JK\cdot 2^K)$ time (refer to Table~\ref{tab:alg} for the step-by-step algorithm). For an alternative algorithm using Boolean algebra, see \citeA{Tatsuoka1991}.
Note that our algorithm results in equivalence classes labeled by their smallest member, which shall be called the \emph{minimal representative}.
The minimal representative has additional meaning as the attribute requirements of the corresponding ideal response vector and is therefore the most preferable member for labeling.

\begin{table}[ht]
\begin{center}
\caption{Algorithm for partitioning an attribute profile space }
\label{tab:alg}
\begin{tabular}{rp{3.75in}}
\toprule
Step & Procedure\\
\midrule
Input: & A $J\times K$ Q-matrix $Q$. \\
\noalign{\medskip}
(0) & (optional) Remove items with duplicate attribute requirements\\
(1) & List all $2^K$ attribute profiles $\aalpha$.\\
(2) & Find the ideal response vector $\xxi(Q,\aalpha)$ for each $\aalpha$.\\
(3) & Do a lexicographic (alphabetic) sort of the ideal response vectors.\\
(4) & Check whether each successive profile has the same ideal response vector as the previous profile. If so, $\aalpha$ is the first member of a new equivalence class $[\aalpha]$.  Else, $\aalpha$ is part of the current equivalence class.\\
\noalign{\medskip}
Output: & A list of equivalence classes $[\aalpha]$ and their members.\\
\bottomrule
\end{tabular}
\end{center}
\end{table}

\begin{table}[ht]
\begin{center}
\caption{Generating the partition associated with the Q-matrix $Q_3$.}
\label{tab:ex}
\begin{tabular}{cccccccccc}
\toprule
$Q_3$&&$\aalpha$ & $\xxi(Q_3,\aalpha)$&& $\aalpha$ & $\xxi(Q_3,\aalpha)$ &&$[\aalpha]$\\
\cmidrule(lr){1-1} \cmidrule(lr){3-4} \cmidrule(lr){6-7} \cmidrule(lr){9-9}
\multirow{8}{*}{$\begin{pmatrix}1&0&0\\1&1&0\\0&1&1\end{pmatrix}$} & \multirow{8}{*}{$\underrightarrow{\text{ (1),(2)}}$}&000&000&\multirow{8}{*}{$\underrightarrow{(3)}$}&000&000&\multirow{8}{*}{$\underrightarrow{\text{(4)}}$}&[000]\\
&&100&100&&010&000\\
&&010&000&&001&000\\
&&001&000&&011&001&&[011]\\
&&110&110&&100&100&&[100]\\
&&101&100&&101&100\\
&&011&001&&110&110&&[110]\\
&&111&111&&111&111&&[111]\\
\bottomrule
\end{tabular}
\flushleft
Steps from Table~\ref{tab:alg} labeled (1), (2), (3), and (4).\end{center}
\end{table}

As seen in Table~\ref{tab:ex}, performing the algorithm on the  $3\times 3$ Q-matrix $Q_3$ from \eqref{eq:Q3} results in five different equivalence classes, each of which is labeled with by its minimal representative: $[000] = \{000,010,001\}$, $[011] = \{011\}$, $[100] = \{100,101\}$, $[110] = \{110\}$, and $[111] = \{111\}$. Note that since the bracket notation may be read as `the equivalence class containing,' it is possible to change the labeling of each equivalence class by choosing any other member as the titular profile: $[000]$, $[010]$, and $[001]$ all refer to the same equivalence class, for example.

\section{Consistent Estimation}\label{sec:est}
We now consider the problem of parameter estimation, specifically that of $\nu_{\aalpha}$, the proportion of the population possessing each attribute profile $\aalpha$. Unless $\nnu$ is assumed known, its consistent estimation has important consequences for respondent classification and exam validity. Unfortunately, when an assessment's Q-matrix is incomplete, it is impossible to consistently estimate $\nnu$. For each equivalence class $[\aalpha]$, let $\nu_{[\aalpha]}$ be the proportion of the population possessing an attribute profile within that equivalence class. Then,
\begin{equation}\label{eq:priorClass}
\nu_{[\aalpha]} = \sum_{\aalpha'\in[\aalpha]}\nu_{\aalpha'}.
\end{equation}
The probability of observing any particular set of data depends only on $\nu_{[\aalpha]}$, since the probability of any response depends only on equivalence class membership, not on the respondent's possession of a specific profile. With an incomplete Q-matrix it is possible to observe populations with different distributions $\nnu^1 \not\equiv \nnu^2$ over the attribute profile space that have identical distributions over the equivalence classes $[\aalpha]$, i.e., $\nu^1_{[\aalpha]} = \nu^2_{[\aalpha]}$ for all $\aalpha \in \{0,1\}^K$, and thus identical response distributions. The phenomenon where different parameter values lead to identical response distributions is generally known as non-identifiability, and it destroys the ability of likelihood-based estimation methods to achieve consistency.

While consistent estimation of $\nu_{\aalpha}$ cannot be achieved, it is possible to consistently estimate the proportion of individuals within each equivalence class $[\aalpha]$.
\begin{theorem}\label{thm:consistent}
Suppose an assessment follows the DINA model, with known Q-matrix $Q$ and item parameters $\sb$ and $\gg$. Let $\nu_{[\aalpha]}$, representing  the proportion of the population possessing an attribute profile $\aalpha' \in [\aalpha]$, be defined as in \eqref{eq:priorClass}, and let the population parameter $\nnu$ be the vector of all $\nu_{[\aalpha]}$. We may write its likelihood as
$$L(\nnu) = p(X|\nnu) = \prod_{i=1}^N p(\xx^i|\nnu) = \prod_{i=1}^N \sum_{[\aalpha]}p(\xx^i|[\aalpha])\nu_{[\aalpha]}.$$
Then the maximum likelihood estimate $\hat\nnu$ of $\nnu$  is consistent as $N\rightarrow \infty$.
\end{theorem}
Consistent estimation of the $\nu_{[\aalpha]}$ is an important result, justifying the results of the following sections. To emphasize the differences in parameter space and procedure, work based on equivalence classes $[\aalpha]$ rather than profiles $\aalpha$ will from hereon be referred to under the name of the Non-Identifiability ADjusted DINA (NIAD-DINA) model.

\section{Marginal Identifiability}\label{sec:single}
We now wish to extend the concept of identifiability to individual attributes.  This is motivated by the fact that, though the presence of multiple profiles in the same equivalence class signals non-identifiability, some individual attributes may still be identifiable within the class. To illustrate, consider the the Q-matrix $Q_3$ from \eqref{eq:Q3} and one of its equivalence classes, $[000] = \{000,010,001\}$. If a profile $\aalpha \in [000]$, then its first component $\alpha_1 = 0$, but the values of $\alpha_2$ and $\alpha_3$ are uncertain. Thus, posterior weight $p([000]|x)$ on this class counts as positive evidence that $\alpha_1 = 0$, but does not help in deciding $\alpha_2$ or $\alpha_3$. This observation motivates the following definition:
\begin{definition}
An attribute is \emph{marginally identifiable} within an equivalence class when either all members of that class possess that attribute or none of them do.
\end{definition}
Define the marginal identifiability indicator $\delta_{[\aalpha],k}$ as follows:
\begin{equation}\label{eq: identifiable}
\delta_{[\alpha],k} = \prod_{\aalpha' \in [\aalpha]} \alpha'_k + \prod_{\aalpha'\in[\aalpha]} (1-\alpha'_k).
\end{equation}
Then, $\delta_{[\aalpha],k} = 1$ when Attribute $k$ is marginally identifiable within equivalence class $[\aalpha]$.
Posterior weight on a class $[\aalpha]$ only provides information about the $k$-th attribute when $\delta_{[\aalpha],k} = 1$; otherwise, there is no information beyond the prior. 


\subsection{Exam Quality: the Marginal Identifiability Rate}\label{sec:eval}

Since non-identifiability is frequently unavoidable with Q-matrix based CDMs, it is important to measure its extent. For a more nuanced view, this is done on a marginal, basis.

Given the proportion $\nu_{\aalpha}$ of each attribute profile $\aalpha$, the proportion of the population for which the $k$-th attribute is marginally identifiable can be quantified by $\zeta_k$, as follows:
\begin{equation}\label{eq:zeta}
\zeta_k = \sum_{\{\aalpha: \delta_{[\aalpha],k} = 1\}}\nu_{\aalpha}.
\end{equation}
Let $\zzeta$ be the vector of all $\zeta_k$. Then $\zzeta$ is the proportion of the population for which each attribute is marginally identifiable, i.e., the marginal identifiability rate.

Oftentimes $\nu_{\aalpha}$, and thus $\zzeta$, is unknown. Under the conditions of Theorem~\ref{thm:consistent}, $\zzeta$ can be consistently estimated by its maximum likelihood estimator $\hat\zzeta$.
\begin{proposition} \label{prop:marg}
Suppose an assessment follows the DINA model, with known Q-matrix $Q$ and item parameters $\sb$ and $\gg$. Let $\hat\nu_{[\aalpha]}$ be the MLE estimate of $\nu_{[\aalpha]}$. Then
\begin{equation}\label{eq:propID_hat}
\hat\zeta_k = \sum_{\{[\aalpha]: \delta_{[\aalpha],k} = 1\}}\hat\nu_{[\aalpha]},\ k = 1,\ldots, K.
\end{equation}
is consistent as $N\rightarrow \infty$.
\end{proposition}
The consistency of $\hat\zzeta$ is a direct consequence of the consistency of $\hat\nu_{[\aalpha]}$ in Theorem~\ref{thm:consistent}. We thus obtain a very reasonable measure of exam quality, in terms of the proportion of the population for which each attribute is marginally identifiable.

\section{Classification}\label{sec:class}

Non-identifiability has potentially serious effects on respondent classification. Classification is generally conducted based on the posterior distribution $p(\aalpha|\xx) \propto p(\xx|\aalpha) p(\aalpha)$. Recall that profiles in the same equivalence class have the same likelihood. Thus, the posterior will simply be a reflection of the prior, without any added information from the data \cite{deCarlo2011}. In fact, within an equivalence class the posteriors are proportional to the priors. Given a prior $p(\aalpha)$, for any $\aalpha^1,\aalpha^2 \in [\aalpha]$, 
$$
\frac{p(\aalpha^1|\xx)}{p(\aalpha^2|\xx)} 
= \frac{p(\xx|\aalpha^1|)p(\aalpha^1)}{p(\x|\aalpha^2)p(\aalpha^2)} 
= \frac{p(\xx|[\aalpha])p(\aalpha^1)}{p(\xx|[\aalpha])p(\aalpha^2)} 
= \frac{p(\aalpha^1)}{p(\aalpha^2)}.
$$
Posteriors are often calculated by maximizing the marginal maximum likelihood $L(\nnu,\sb,\gg)$ via the E-M algorithm \cite{Haertel,dela2009, Rupp}. Then, since all vectors $\nnu$ with identical weights on each class $\nu_{[\aalpha]}$ have identical likelihoods, any $\nnu$ achieving the maximizing $\nu_{[\aalpha]}$ may result.  The values chosen are determined by the starting values, which have little validity for classification.

When the posterior is sensitive to the prior, it is important to work with $p([\aalpha])$, which can be estimated consistently, rather than $p(\aalpha)$. Thus classification here will be conducted based on $p([\aalpha]|\xx \propto p(\xx|[\aalpha])p([\aalpha])$ instead of the usual posterior. This calculation does not require a separate fitting of the model, since
\begin{equation}
\label{eq:postClass}
p([\aalpha]|\xx)
= \frac{p(\xx|[\aalpha])\nu_{[\aalpha]}  }{p(\xx)}
= \frac{p(\xx|[\aalpha])\sum_{\aalpha'\in[\aalpha]}\nu_{\aalpha'}}{p(\xx)}
= \sum_{\aalpha'\in[\aalpha]}p(\aalpha'|\xx)  
\end{equation}
From this posterior, we then define
\begin{eqnarray}
p_k^{\min}(\xx) & = & \sum_{[\aalpha]:\alpha_k=1,d_{[\aalpha],k}=1} p([\aalpha]|\xx),\label{eq:pmin} \\
p_k^{\max}(\xx) & = & p_k^{\min}(\xx) +\sum_{[\aalpha]:d_{[\aalpha],k}=0} p([\aalpha]|\xx),\label{eq:pmax}
\end{eqnarray}
where $\delta_{[\aalpha],k}$ is the marginal identifiability indicator defined in \eqref{eq: identifiable}. 
Classification follows from the fact that, depending upon the specific hyperprior on $\nnu$ or starting point of the E-M algorithm, the DINA model may produce marginal posterior probabilities of mastery $p(\alpha_k=1|\xx)$ anywhere in the range $[p_k^{\min}(\xx), p_k^{\max}(\xx)]$.
Thus, it is only appropriate to conclude that $\alpha_k = 1$ when $p_k^{\min}(\xx)$ is high, or that $\alpha_k = 0$ when $p_k^{\max}(\xx)$ is low. A natural cutoff for both is 0.5, but it may be adjusted as necessary. This classification method, from hereon referred to as the NIAD-DINA classification algorithm, accounts for both uncertainty in the prior and uncertainty caused by slipping and guessing. It is summarized in Table~\ref{tab:algClass}.

\begin{table}[ht]
\begin{center}
\caption{ NIAD-DINA classification algorithm}
\label{tab:algClass}
\begin{tabular}{rp{3in}c}
\toprule
Step&Procedure &q.v.\\
\midrule
Input: & Q-matrix $Q = (q_{jk})_{J\times K}$, data $X = (x_{ik})_{N\times J}$.\\
\noalign{\medskip}
(1) & Fit the model to produce $p(\aalpha|\xx)$.\\
(2) & Partition the attribute profile space. & Table~\ref{tab:alg}\\
(3) & Calculate the marginal identifiability vector $\delta_{[\aalpha]}$. &  \eqref{eq: identifiable}\\
(4) & Sum posteriors $p(\aalpha|x)$ for $p([\aalpha]|\xx)$. & \eqref{eq:postClass}\\
(5) & Calculate $p_k^{\min}(\xx)$ and $p_k^{\max}(\xx)$ for every $k,\xx$. & \eqref{eq:pmin}, \eqref{eq:pmax}\\
(6) & Classify: \\
& \hspace{.25in}If $p_k^{\min}>0.5$, then $\hat\alpha_k = 1$. \\
&\hspace{.25in}If $p_k^{\max} < 0.5$, then $\hat\alpha_k = 0$. \\
&\hspace{.25in}Else, $\hat\alpha_k = *$ (unclassified).\\
\noalign{\medskip}
Output: & Classifications $\hat\alpha^i_k  \in \{0,1,*\}$  for all $i,k$.\\
\bottomrule
\end{tabular}
\end{center}
\end{table}

\section{Extensions}\label{sec:exts}
\subsection{Variants of the DINA}
The methodology of partitioning in Section \ref{sec:part} applies to any model where the presence of a difference in the ideal response pattern fully determines the presence of a difference in the likelihood function. Included among these models are all variants of the DINA model listed in Section~\ref{sec:DINAvar}. Since these variants are, in essence, a restriction on the space of $\nu$, the  consistency result for $\nu_{[\aalpha]}$ in Section~\ref{sec:est}, along with all the following results, holds when the model is in fact correct. However, if the true $\nu_{[\aalpha]}$ do not fall under the set of values consistent with the restriction on the parameter space, then even estimates of $\nu_{[\aalpha]}$ will no longer be consistent. Thus, large differences in the $\hat\nu_{[\aalpha]}$ calculated under restricted models from those calculated under the NIAD-DINA model are symptomatic of model misspecification, and may imply that the DINA variant chosen is overly restrictive on the prior. Goodness-of-fit measures such as the AIC and BIC will reflect lack of fit appropriately if the saturated model has the correct number of parameters $2J + L$, rather than $2J + 2^K-1$.

\subsection{The DINO Model}
The DINO model also specifies item and attribute relationships using a Q-matrix, but the ideal responses are calculated as
$$\xi_j(Q,\aalpha)  = 1- \prod_{k=1}^K (1-\alpha_k)^{q_{jk}} = \textbf{1}(\alpha_k = q_{jk}=1\text{ for some }k).$$
As in the DINA model, the response probabilities are functions of item parameters $s_j = P(X_j = 0|\xi_j = 1)$ and $g_j = P(X_j = 1|\xi_j = 0)$. Under the DINA model, ideal responses are \emph{correct} when the respondent possesses all required attributes; under the DINO model, ideal responses are \emph{incorrect} when the respondent \emph{does not} possess all required attributes. Thus, for responses $X$ following the DINO model, the reverse responses $1-X$ follow the DINA model (with a reversed interpretation of the attribute profile vectors). This dualism implies that all the results of this paper apply to the DINO model.

\section{Results}\label{sec:results}
\subsection{Simulation Results}
We first demonstrate the procedures on simulated data. Responses are generated for $N = 5000$ resondents taking an assessment with $J = 6$ items measuring $K = 3$ distinct attributes. The respondents' mastery or nonmastery of the measured attributes 
is randomly generated according to the probability $p_{sim}(\aalpha)$ of each profile $\aalpha \in\{0,1\}^3$, as listed in Table~\ref{tab:p}. 
\begin{table}[ht]
\centering
\caption{Population proportions of each attribute profile}
\label{tab:p}
\begin{tabular}{lcccccccc}
\toprule
&\multicolumn{8}{c}{ $\aalpha$ }\\
\cmidrule(lr){2-9}
& 000 & 001 & 010 & 011 & 100 & 101 & 110 & 111 \\
\midrule
$p_{sim}$ & 0.27 & 0.00 & 0.01 & 0.04 & 0.10 & 0.16 & 0.20 & 0.21\\
\bottomrule
\end{tabular}
\end{table}
The responses themselves follow the DINA model according to the Q-matrix $Q_{sim}$ with slipping $s_{sim}$ and guessing $g_{sim}$ as shown in Table~\ref{tab:simParam}.
\begin{table}[ht]\footnotesize
\centering
\caption{Q-matrix, slipping, and guessing for simulated data.}
\begin{tabular}{cccc}
\toprule
Item ($j$) & Attribute vector ($\qq^j$) &Slipping ($s_j$) &Guessing ($g_j$) \\
\midrule
1. & 100 & 0.14 & 0.10 \\
2. & 110 & 0.12 & 0.15 \\
3. & 011 & 0.18 & 0.18 \\
4. & 100 & 0.17 & 0.18 \\
5. & 110 & 0.08 & 0.06 \\
6. & 011 & 0.05 & 0.06 \\
 \bottomrule
 \end{tabular}
 \label{tab:simParam}
\end{table}

The Q-matrix $Q_{sim}$ is incomplete, and the resulting instability in the posterior becomes clear once the data is fitted multiple times. As an example, the posterior probabilities of each attribute profile given the zero response vector $\bfzero = (0,0,\ldots, 0)$ are summarized in Table~\ref{tab:simPost0_profiles}. 
\begin{table}[ht]
\centering
\caption{Posterior probabilities given zero correct responses, $p(\aalpha|\xx = \bfzero)$}
\begin{tabular}{llcccccccc}\toprule
&&\multicolumn{8}{c}{ $\aalpha$ }\\
 \cmidrule(lr){3-10}
&& 000 & 001 & 010 & 011 & 100 & 101 & 110 & 111 \\
\midrule
\multicolumn{2}{l}{truth} \\
&& 0.91 & 0.02 & 0.05 & 0.00 & 0.01 & 0.01 & 0.00 & 0.00\\
\multicolumn{2}{l}{minimums} \\
&  DINA & 0.01 & 0.02 & 0.03 & 0.00 & 0.00 & 0.00 & 0.00 & 0.00\\
& HO-DINA & 0.13 & 0.09 & 0.03 & 0.00 & 0.01 & 0.01 & 0.00 & 0.00\\
&RHO-DINA & 0.55 & 0.11 & 0.29 & 0.00 & 0.02 & 0.01 & 0.00 & 0.00\\
&ind-DINA & 0.29 & 0.24 & 0.37 & 0.00 & 0.02 & 0.02 & 0.00 & 0.00\\
\multicolumn{2}{l}{maximums} \\
&DINA & 0.71 & 0.86 & 0.56 & 0.00 & 0.02 & 0.03 & 0.00 & 0.00\\
&HO-DINA & 0.62 & 0.81 & 0.71 & 0.00 & 0.02 & 0.02 & 0.00 & 0.00\\
&RHO-DINA & 0.58 & 0.13 & 0.30 & 0.00 & 0.02 & 0.01 & 0.00 & 0.00\\
&ind-DINA & 0.31 & 0.28 & 0.40 & 0.01 & 0.03 & 0.02 & 0.00 & 0.00\\
\noalign{\medskip}
\bottomrule
\end{tabular}
\flushleft
Note: Minimum and maximum values of the posterior $p(\aalpha|\xx = 0)$, as  generated over ten runs of the (random start) E-M algorithm. 
\label{tab:simPost0_profiles}
\end{table}
Here, the DINA, HO-DINA, and RHO-DINA are overparameterized and produce a wide range of results for Profiles $[000]$, $[001]$, and $[010]$. The slight variability in the ind-DINA estimates is a numerical artifact. While the ind-DINA does not suffer from nonidentifiability, it still does not give accurate estimates in this case since the model assumptions are incorrect.

Partitioning the attribute profile space as directed by Table~\ref{tab:alg} produces the five equivalence classes listed in Table~\ref{tab:simPart}, two of which have multiple members.  The table also reports the marginal identifiability vector $\delta_{[\aalpha]}$ for each class. Note that since Items 1 and 4 are devoted to Attribute 1, it is always marginally identifiable and $\delta_{[\aalpha],1}\equiv 1$. It is also clear that nonidentifiability most seriously affects Attribute 3, which is marginally non-identifiable for members of both [000] and [100]. Finally, Table ~\ref{tab:simPart} also reports E-M estimates of the proportion of respondents in each class under the DINA and several variants, along with the true proportion. Note the accuracy of the DINA estimates, which are consistent, and the inaccuracy of the ind-DINA estimates due to model misfit.
\begin{table}[ht]
\centering
\caption{Equivalence classes, along with their class sizes, true and maximum likelihood probabilities, and marginal identifiability vectors.}
\begin{tabular}{cccccccc}\toprule
&&\multicolumn{5}{c} {$\nu_{[\aalpha]}$}&\\
\cmidrule(lr){3-7}
$[\aalpha]$ & Size & True & DINA & HO-DINA & RHO-DINA & ind-DINA &$\delta_{[\aalpha]}$\\
\midrule
 $[000]$ & 3 & 0.29 & 0.30 & 0.29 & 0.29 & 0.22 & 100\\
 $[100]$ & 2 & 0.26 & 0.26 & 0.27 & 0.26 & 0.31 & 110\\
 $[011]$ & 1 & 0.04 & 0.04 & 0.04 & 0.04 & 0.08 & 111\\
 $[110]$ & 1 & 0.20 & 0.20 & 0.19 & 0.20 & 0.20 & 111\\
 $[111]$ & 1 & 0.21 & 0.21 & 0.21 & 0.21 & 0.18 & 111\\
\bottomrule
\end{tabular}
\label{tab:simPart}
\end{table}

We now consider variability in the marginal posterior probabilities $p(\alpha_k = 1|\xx)$. Table~\ref{tab:simPost0_alphaK} gives the sample range of $p(\alpha_k = 1|\xx)$ after ten runs of the E-M algorithm, in addition to  the  theoretical range. Note the large theoretical ranges for $p(\alpha_k=1|\xx = \bfzero), k = 2,3$. 
\begin{table}[ht]
\centering
\caption{Variability in $p(\alpha_k=1|\xx=\bfzero)$, the marginal posterior given the zero response vector.}
\begin{tabular}{lccc}\toprule
&\multicolumn{3}{c}{ $k$ }\\
 \cmidrule(lr){2-4}
& 1 & 2 & 3\\
\midrule
sample min  &  0.03 &   0.03 &   0.04\\
$p_k^{min}(\bfzero)$  &  0.03 &   0.00  &  0.00\\
  \noalign{\medskip}
sample max  &  0.03   & 0.56  &  0.89\\
$p_k^{max}(\bfzero)$   & 0.03   & 0.97 &   1.00\\
\bottomrule
\end{tabular}
\flushleft
Note: Probabilities calculated by fitting the DINA model over ten runs of E-M algorithm with random starts.
\label{tab:simPost0_alphaK}
\end{table}

Classification was conducted on a marginal basis, based on $p(\alpha_k=1|\xx)$, under each of the models. In addition, NIAD-DINA classification was performed (see Table~\ref{tab:algClass}). Marginal misclassification rates $p(\hat\alpha_k \neq \alpha_k)$ are compared in Table~\ref{tab:simMisclass}. Note that NIAD-DINA classification results in unclassified individuals; for example, $\hat\aalpha = (0**)$ for those with the zero response vector. This rate is also listed in Table~\ref{tab:simMisclass}. The DINA and HO-DINA are overparameterized and the misclassification rate for Attribute 3 may reach over 40\% in both models. The ind-DINA also performs poorly, but due to an overly restricted parameter space rather than nonidentifiability. Adjusting classification under the DINA to account for nonidentifiability according to the method described in Section~\ref{sec:class} solves both these issues. It may leave a large proportion of individuals unclassified, but this is a necessary consequence of the assessment design. Classifying these individuals would require further assumptions beyond the model.
\begin{table}[h]
\centering
\caption{Marginal misclassification rates under a variety of models.}
\begin{tabular}{cccccc}\toprule
&\multicolumn{5}{c}{ Model }\\
 \cmidrule(l){2-6}
k & DINA & HO-DINA & RHO-DINA & ind-DINA& NIAD-DINA\\
\midrule
1 & 0.07	       & 0.07          & 0.07 & 0.09 & 0.07 (0.00)\\
2 & 0.07-0.32 & 0.07-0.32 & 0.07 & 0.26 & 0.04 (0.32)\\
3 & 0.19-0.44 & 0.19-0.43 & 0.20 & 0.21 & 0.04 (0.56)\\
\bottomrule
\end{tabular}
\flushleft
Note: Range over 10 runs reported for overparameterized models. All cut-offs equal to 0.5. The proportion of respondents left unclassified under the NIAD-DINA is displayed within parentheses.
\label{tab:simMisclass}
\end{table}


In addition to controlling misclassification errors, we may also evaluate the quality of the assessment by measuring the marginal identifiability rate $\zzeta$ defined in \eqref{eq:zeta}. Table~\ref{tab:simID} shows both true and estimated values for $\zzeta$. Note once again that nonidentifiability affects Attribute 3 more severely than it does Attribute 2. In addition, estimates are generally accurate, except in the case of the ind-DINA, which suffers from lack of fit.
\begin{table}[h]
\centering
\caption{True and estimated values for $\zzeta$, marginal identifiability rate.}
\begin{tabular}{cccccc}\toprule
&&\multicolumn{4}{c}{ Model }\\
\cmidrule(l){3-6}
k & true & DINA & HO-DINA & RHO-DINA & ind-DINA\\
\midrule
1 & 1.00 & 1.00 & 1.00 & 1.00 & 1.00\\
2 & 0.71 & 0.70 & 0.71 & 0.71 & 0.78\\
3 & 0.44 & 0.45 & 0.45 & 0.45 & 0.47\\
\bottomrule
\end{tabular}
\label{tab:simID}
\end{table}


In terms of model selection, reducing the number of parameters for the DINA model to $2M+L$ from the original $2M + 2^K$ reduces the comparative advantage of the restricted models. In Table~\ref{tab:simAIC}, the AIC value of the RHO-DINA barely edges out that of the DINA with identifiability adjustment.
\begin{table}[h]
\centering
\caption{AIC and BIC for the DINA, RHO-DINA, and ind-DINA.}
\begin{tabular}{lccc}\toprule
& parameters & AIC & BIC\\
 \midrule
NIAD-DINA        & 17 & 32862.8 & 32973.6 \\
RHO-DINA    & 16 &32861.2 & 32965.5 \\
ind-DINA & 15 & 32995.9 & 33093.6 \\
\bottomrule
\end{tabular}
\label{tab:simAIC}
\end{table}

\subsection{A Fractions Data Example}
We now turn to the widely analyzed fraction subtraction data set of \citeA{Tatsuoka1990}. It is composed of the twenty items listed in Table~\ref{tab:frac_jtems}.
\begin{table}[h]
\caption{Items from the fraction subtraction data set \protect\cite{Tatsuoka1990}.}
\begin{tabular}{cccccc}
\toprule
No. & Item & No. & Item & No. & Item\\
\midrule
\,1. & $\nicefrac{5}{3}-\nicefrac{3}{4}$ & 
\,8. & $\nicefrac{2}{3}-\nicefrac{2}{3}$ & 
15. & $2 - \nicefrac{1}{3}$ \\
\,2. & $\nicefrac{3}{4} - \nicefrac{3}{8}$ & 
\,9. & $3\,\nicefrac{7}{8}-2$ & 
16. & $4\,\nicefrac{5}{7} - 1\,\nicefrac{4}{7}$\\
\,3. & $\nicefrac{5}{6}-\nicefrac{1}{9}$ & 
10. & $4\,\nicefrac{4}{12} - 2\,\nicefrac{7}{12}$ & 
17. & $7\,\nicefrac{3}{5} - \nicefrac{4}{5}$\\
\,4. & $3\,\nicefrac{1}{2}-2\,\nicefrac{3}{2}$ & 
11. & $4\,\nicefrac{1}{3}-2\,\nicefrac{4}{3}$ & 
18. & $4\,\nicefrac{1}{10} - 2\,\nicefrac{8}{10}$\\
\,5. & $4\,\nicefrac{3}{5} - 3\,\nicefrac{4}{10}$ & 
12. & $\nicefrac{11}{8}-\nicefrac{1}{8} $& 
19. & $4 - 1\,\nicefrac{4}{3}$\\
\,6. & $\nicefrac{6}{7}-\nicefrac{4}{7}$ & 
13. & $3\,\nicefrac{3}{8}-2\,\nicefrac{5}{6}$ & 
20. & $4\,\nicefrac{1}{3} - 1\,\nicefrac{5}{3}$\\
\,7. & $3-2\,\nicefrac{1}{5}$ & 
14. & $3\,\nicefrac{4}{5}-3\,\nicefrac{2}{5}$ \\
 \bottomrule
 \end{tabular}
 \label{tab:frac_jtems}
\end{table}
The Q-matrix in Table~\ref{tab:Qfrac} comes from \citeA{dela}, and specifies the following eight attributes: $\alpha_1 = $ convert a whole number to a fraction; $\alpha_2 = $ separate a whole number from a fraction; $\alpha_3 = $ simplify before subtracting; $\alpha_4 = $ find a common denominator; $\alpha_5 = $ borrow from whole number part; $\alpha_6 = $ column borrow to subtract the second numerator from the first; $\alpha_7 = $ subtract numerators; and $\alpha_8 = $ reduce answers to simplest form. 

\begin{table}[ht]\footnotesize
\centering
\caption{Q-matrix from \protect\citeA{dela}.}
\begin{tabular}{cccccc}
\toprule
Item & Attributes ($\qq^j$) &Item &Attributes ($\qq^j$) &Item &Attributes ($\qq^j$) \\
\midrule
\ 1. & 00010110 & \ 8. & 00000010 & 15. & 10000010 \\
\ 2. & 00010010 & \ 9. & 01000000 & 16. & 01000010\\
\ 3. & 00010010 & 10. & 01001011 & 17. & 01001010\\
\ 4. & 01101010 & 11. & 01001010 & 18. & 01001110\\
\ 5. & 01010011 & 12. & 00000011 & 19. & 11101010\\
\ 6. & 00000010 & 13. & 01011010 & 20. & 01101010\\
\ 7. & 11000010 & 14. & 01000010   \\
 \bottomrule
 \end{tabular}
 \label{tab:Qfrac}
\end{table}

As pointed out by \citeA{deCarlo2011}, this assessment exemplifies the identifiability issues of the DINA model. While Attributes 2 and 7 have items dedicated solely to them, all other attributes appear only in combination. In fact, Attribute 3 only appears in Item 4, in conjunction with Attributes 2, 5, and 7. Attribute 7 is required for all items except one, making it difficult to draw conclusions about other attributes when it has not been mastered. Table~\ref{tab:frac_post0} displays the marginal posterior probabilities of mastery for each attribute, given the zero response vector. The posterior displayed for the DINA is just one possible output of the E-M algorithm for this data; meanwhile, note the high probabilities of mastery under the ind-DINA model. Common sense dictates that something is out of place when the analysis states that students with a score of zero cannot subtract numerators, but can do everything else, from finding a common denominator to borrowing to reducing to simplest form.
 \begin{table}[!ht]
 \begin{center}
 \caption{Marginal posterior probabilities of mastery given the zero response vector, $p(\alpha_k=1|\xx=\bfzero)$}
 \begin{tabular}{lcccccccc}
 \toprule
 &\multicolumn{8}{c}{k}\\
 \cmidrule(l){2-9}
& 1 & 2 & 3 & 4 & 5 & 6 & 7 & 8\\
\midrule
DINA		& 0.50 & 0.08 & 0.50 & 0.52 & 0.53 & 0.41 & 0.00 & 0.59\\
HO-DINA	& 0.00 & 0.13 & 0.31 & 0.05 & 0.02 & 0.30 & 0.00 & 0.25\\
RHO-DINA	& 0.02 & 0.13 & 0.12 & 0.05 & 0.02 & 0.25 & 0.00 & 0.18\\
ind-DINA		& 0.74 & 0.86 & 0.96 & 0.86 & 0.75 & 0.98 & 0.00 & 0.94\\
 \bottomrule
 \end{tabular}
 \label{tab:frac_post0}
 \end{center}
\end{table}

\begin{table}[!ht]
\centering
\caption{Multiple-member equivalence classes, along with their class sizes, maximum likelihood probabilities, and marginal identifiability vectors.}
\begin{tabular}{ccccccc}\toprule
&&\multicolumn{4}{c} {$\nu_{[\aalpha]}$}&\\
\cmidrule(lr){3-6}
$[\aalpha]$ & Size & DINA & HO-DINA & RHO-DINA & ind-DINA &$\delta_{[\aalpha]}$\\
\midrule
$[00000000]$ & 64 & 0.15 & 0.12 & 0.12 & 0.02 & 01000010\\
$[01000000]$ & 64 & 0.04 & 0.06 & 0.06 & 0.31 & 01000010\\
$[00000010]$ &   8 & 0.01 & 0.02 & 0.03 & 0.00 & 11010011\\
$[10000010]$ &   8 & 0.00 & 0.00 & 0.00 & 0.00 & 11010011\\
$[00000011]$ &   8 & 0.02 & 0.03 & 0.02 & 0.00 & 11010011\\
$[10000011]$ &   8 & 0.00 & 0.00 & 0.00 & 0.00 & 11010011\\
$[01000010]$ &   4 & 0.03 & 0.04 & 0.04 & 0.00 & 11011011\\
$[01000011]$ &   4 & 0.11 & 0.09 & 0.08 & 0.00 & 11011011\\
$[11000010]$ &   4 & 0.00 & 0.00 & 0.00 & 0.00 & 11011011\\
$[11000011]$ &   4 & 0.01 & 0.01 & 0.02 & 0.01 & 11011011\\
$[00010010]$ &   4 & 0.00 & 0.00 & 0.00 & 0.00 & 11010111\\
$[10010010]$ &   4 & 0.00 & 0.00 & 0.00 & 0.00 & 11010111\\
$[00010011]$ &   4 & 0.00 & 0.00 & 0.00 & 0.00 & 11010111\\
$[10010011]$ &   4 & 0.00 & 0.00 & 0.00 & 0.00 & 11010111\\
$[00010110]$ &   4 & 0.02 & 0.00 & 0.00 & 0.00 & 11010111\\
$[10010110]$ &   4 & 0.00 & 0.00 & 0.00 & 0.00 & 11010111\\
$[00010111]$ &   4 & 0.00 & 0.01 & 0.01 & 0.01 & 11010111\\
$[10010111]$ &   4 & 0.00 & 0.00 & 0.00 & 0.03 & 11010111\\
$[01010010]$ &   2 & 0.00 & 0.00 & 0.00 & 0.00 & 11011111\\
$[11010010]$ &   2 & 0.00 & 0.00 & 0.00 & 0.00 & 11011111\\
$[01010011]$ &   2 & 0.01 & 0.01 & 0.00 & 0.00 & 11011111\\
$[11010011]$ &   2 & 0.00 & 0.00 & 0.00 & 0.00 & 11011111\\
$[01010110]$ &   2 & 0.00 & 0.01 & 0.01 & 0.00 & 11011111\\
$[11010110]$ &   2 & 0.00 & 0.00 & 0.00 & 0.01 & 11011111\\
$[01010111]$ &   2 & 0.05 & 0.06 & 0.06 & 0.03 & 11011111\\
$[11010111]$ &   2 & 0.06 & 0.06 & 0.06 & 0.09 & 11011111\\
\bottomrule
\end{tabular}
\label{table:fracParts}
\end{table}
With eight attributes in the Q-matrix, there are a total of 256 possible attribute profiles. They can be divided into just 58 different equivalence classes by the partitioning algorithm, 32 of them containing a single identifiable element. The 26 multiple-profile equivalence classes are listed in Table~\ref{table:fracParts}, which also displays their class sizes, maximum likelihood probabilities, and marginal identifiability vectors. Within these multiple-profile equivalence classes, Attributes 2 and 7 are always marginally identifiable, while Attribute 3 is never so; this is natural considering our previous observations about the Q-matrix. Profiles within the largest classes contain many zeroes, since under the DINA model non-identifiability affects a particular attribute only for respondents who do not possess other attributes used in combination with that attribute. Also note that the ind-DINA shows signs of model misspecification, since its estimates $\hat\nu_{[\aalpha]}$ deviate strongly from the estimates derived from the other models. 

Table~\ref{table:fracEval} shows the estimated marginal identifiability rates, $\hat\zzeta$. At the low end, $\hat\zeta_3 = 0.48$, bringing into question the ability of this assessment to measure mastery of Attribute 3. Attribute 6 does only slightly better, with $\hat\zeta_6 = 0.64$. Note that Attribute 6 is only utilized in Items 1 and 18; in both cases it appears in conjunction with at least two other attributes.
 \begin{table}[!ht]
 \begin{center}
 \caption{Estimated marginal identifiability rates $\zeta_k$.}
 \begin{tabular}{lcccccccc}
 \toprule
 &\multicolumn{8}{c}{k}\\
 \cmidrule(l){2-9}
& 1 & 2 & 3 & 4 & 5 & 6 & 7 & 8\\
\midrule
DINA           	& 0.81 & 1.00 & 0.48 & 0.81 & 0.75 & 0.64 & 1.00 & 0.81\\
HO-DINA 	& 0.82 & 1.00 & 0.47 & 0.82 & 0.75 & 0.64 & 1.00 & 0.82\\
RHO-DINA 	& 0.82 & 1.00 & 0.48 & 0.82 & 0.75 & 0.63 & 1.00 & 0.82\\
ind-DINA        	& 0.66 & 1.00 & 0.47 & 0.66 & 0.62 & 0.64 & 1.00 & 0.66\\
 \bottomrule
 \end{tabular}
 \label{table:fracEval}
 \end{center}
\end{table}

Finally, it is useful to consider how the reduction of the parameter space for the DINA model, based on identifiability, affects model selection by AIC and BIC. In particular, the AIC no longer prefers the ind-DINA to the full DINA model once the reduced parameter space has been applied. The BIC, which will generally choose sparser models than the AIC, still reports lower values for the ind-DINA, but the comparison is much tighter.
\begin{table}[!ht]
\centering
\caption{AIC and BIC for the DINA, RHO-DINA, and ind-DINA.}
\begin{tabular}{lccc}\toprule
& parameters & AIC & BIC\\
 \midrule
DINA	        	& 296 & 9397.0 & 10665.2 \\
NIAD-DINA        	& \ 98 & 9001.0 & \ 9420.9 \\
HO-DINA    	& \ 56 & 8959.7 & \ 9199.6 \\
RHO-DINA    	& \ 49 & 8961.9 & \ 9171.9 \\
ind-DINA 		& \ 48 & 9208.3 & \ 9413.9 \\
\bottomrule
\end{tabular}
\label{tab:fracAIC}
\end{table}

\section{Discussion}
In general, it is difficult to obtain a complete Q-matrix. Oftentimes, due to the demands of practicality, assessments must involve items that require a combination of skills. Using the tools discussed in this paper, it is possible to determine the extent to which nonidentifiability affects classification and estimation under the DINA model. Marginal identifiability rates $\zzeta$, which can be estimated consistently, provide an overall measure of the extent of non-identifiability; meanwhile, NIAD-DINA classification takes marginal identifiability into consideration in order to control classification errors that are otherwise quite sensitive to the prior information. The results here suggest that when designing items to test a particular attribute, if using a combination of skills is unavoidable, it is best to combine that attribute with basic attributes mastered by a large proportion of the population. After all, it is only impossible to recover information about a particular attribute when the respondent does not posess one or more of the other attributes tested by the same item. 

Many of the methods currently in use may resolve idenifiability issues by enforcing restrictions on the attribute profile space. Variants of the DINA such as the ind-DINA, HO-DINA, and RHO-DINA accomplish this by specifying a structure and a prior on the probabilities $p(\aalpha)=\nu_{\aalpha}$. Although this may eliminate non-identifiability and create a unique global maximum for the likelihood, model misspecification becomes a risk. Thus, careful comparison of these variants to the NIAD-DINA becomes important.

\vspace{\fill}\pagebreak

\appendix
\renewcommand{\theequation}{A\arabic{equation}}
\setcounter{equation}{0}

\section{Proofs}

\begin{proof}[Proof of Proposition~\ref{thmIdeal}]
Suppose $\xi^1_j = \xi^2_j$ for all $j$ such that $ 1-s_j \neq g_j$. 
If $1-s_j = g_j$, then the response distribution for item $j$ does not depend on $\xxi$:
\begin{align*}
P( x_j|\xxi,s,g) = &(1-s_j)^{\xi_jx_j} g_j^{(1-\xi_j)x_j}s_j^{\xi_j(1-x_j)}(1- g_j)^{(1-\xi_j)(1-x_j)} \\
= &(1-s_j)^{x_j} s_j^{1-x_j} = g_j^{x_j}(1-g_j)^{1-x_j}.
\end{align*}
Thus, for every $\xx \in \{0,1\}^J$, 
\begin{align*}
P(\xx|\xxi^1,s,g) 
=&\prod_{j=1}^m P(x_j|\xi^1_j,s_j,g_j)\\
=&\prod_{\{i:1-s_j=g_j\}} P( x_j|\xi^1_j,s_j,g_j)
\prod_{\{i:1-s_j \neq g_j\}} P(x_j|\xi^1_j,s_j,g_j)\\
=&\prod_{\{i:1-s_j=g_j\}} P( x_j|\xi^2_j,s_j,g_j)
\prod_{\{i:1-s_j \neq g_j\}} P( x_j|\xi^2_j,s_j,g_j)
= P(\x|\xxi^2,s,g)
\end{align*}
and $\aalpha^1$ cannot be separated from $\aalpha^2$.
 
Now suppose that $\xi^1_j \neq \xi^2_j$ for some $j$ such that $ 1-s_j \neq g_j$. Then
\begin{align*}
P( x_j = 1|\xxi^1,s,g) &= (1-s_j)^{\xi^1_j} g_j^{1-\xi^1_j}  \neq g_j^{\xi^1_j}(1-s_j)^{1-\xi^1_j} \\
&=g_j^{1-\xxi^2}(1-s_j)^{\xi^2_j} = P(x_j = 1|\xxi^2,s,g),
\end{align*}
so the response distributions differ.
\end{proof}

\begin{proof}[Proof of Proposition~\ref{prop:complete}]
Suppose that $Q$ is complete. WLOG, for $j = 1, \ldots, K$ let $\qq^j = e_j$. Then, for $j = 1, \ldots, K$, 
$\xi_j(Q,\aalpha) = \alpha_j$ and given any two attribute profiles $\aalpha^1\neq \aalpha^2$, 
$$\xi_{1:K}(Q,\aalpha^1) = \aalpha^1 \neq \aalpha^2 = \xi_{1:K}(Q,\aalpha^2)$$
By Proposition \ref{thmIdeal}, $Q$ separates any $\aalpha^1 \neq \aalpha^2$. 
 
Now suppose that  $\exists k_* \in \{1,\ldots, K\}$ such that $e_{k_*} \not \in R_Q$. WLOG, suppose $k_* = 1$. Consider profiles $\aalpha = e_1$ and $\aalpha' = \bfzero$, the zero column-vector. For each item $j = 1, \ldots, J$, if $q_{j1} = 0$ then 
\begin{align*}
\xi_j(Q,e_1) = \left(1^0\right) \prod_{k\neq 1} 0^{q_{jk}}
 =&  \left(0^0\right) \prod_{k\neq 1} 0^{q_{jk}} =\xi_j(Q,\bfzero).
\end{align*}
Else, $q_{j1} = 1$ and there exists some $k_{**} \neq 1$ such that $q_{jk_{**}} = 1$ and
\begin{align*}
\xi_j(Q,e_1) = \left(0^1\right) \prod_{k\neq k_{**}} [\bfone(k=1)]^{q_{jk}}
 = 0 = \left(0^1\right) \prod_{k\neq k_{**}} 0^{q_{jk}} =\xi_j(Q,\bfzero).
\end{align*}
Thus, $\xxi(Q,e_1) = \xxi(Q,\bfzero)$ and by Proposition \ref{thmIdeal}, attribute profiles $e_1$ and $\bfzero$ cannot be separated.
\end{proof}

\begin{proof}[Proof of Proposition~\ref{prop:equiv}]
The relation `$\sim$' is (i) reflexive: $\xxi(Q,\aalpha) = \xxi(Q,\aalpha) \Rightarrow \aalpha \sim \aalpha$ for all profiles $\aalpha$, (ii) symmetric: if $\aalpha^1 \sim \aalpha^2$, then $\xxi(Q,\aalpha^1) = \xxi(Q,\aalpha^2)$ and $\aalpha^2 \sim \aalpha^1$ for any profiles $\aalpha^1, \aalpha^2$; (iii) transitive: if $\aalpha^1 \sim \aalpha^2$ and $\aalpha^2 \sim \aalpha^3$, then $\xxi(Q,\aalpha^1) = \xxi(Q,\aalpha^2) = \xxi(Q,\aalpha^3)$ and $\aalpha^1 \sim \aalpha^3$ for any profiles $\aalpha^1, \aalpha^2, \aalpha^3$.
\end{proof}

\begin{proof}[Proof of Theorem~\ref{thm:consistent}]
We can write the likelihood  
$$L(\nnu) = p(X|\nnu) = \prod_{i=1}^N p(\xx^i|\nnu)$$
and the log-likelihood as
$$\ell(\nnu) = \sum_{i=1}^N \log p(\xx^i|\nnu) = \sum_{\xx\in \{0,1\}^J} N_{\xx} \log \left(\sum_{[\aalpha]} p\left(\xx\left|[\aalpha]\right.\right)\nu_{[\aalpha]}\right),$$
where $N_{\xx} = \#\{i: \xx^i = \xx\}$.

We first check for identifiability. Suppose there are $L$ distinct equivalence classes partitioning the attribute profile space. The probability of each vector $\xx$ given each class $[\aalpha]$ can be written as a $2^J\times L$ matrix $P = (p_{\xx,[\aalpha]})$, where $p_{\xx,[\aalpha]} = p(\xx|[\aalpha])$. The vector of total probabilities for each response vector $\xx$, for any $\nnu$, can be written as the matrix product $P\nnu$. Thus, we need to show that there is no $\nnu^1,\nnu^2$ s.t.\ $P\nnu^1=P\nnu^2$. Define the vector inequality operation ``$\geq$'' so that $\xx \geq \yy$ iff $x_j \geq y_j$ for all $j$. Let the $T$-matrix be the $2^J\times L$ matrix $T=(t_{\xx,[\aalpha]})$, indexed over all response vectors $\xx$ and equivalence classes $[\aalpha]$, such that $t_{\xx,[\aalpha]} = p(\XX\geq \xx|[\aalpha])$. 
This is a variant of the $T$-matrix used to examine Q-matrix identifiability in \citeA{JLGXZY2011}.
Then $P\nnu^1=P\nnu^2$ iff $T\nnu^1=T\nnu^2$, and the identifiability condition is equivalent to  $T$ being a rank $L$ matrix.

First, suppose that $g \equiv 0$. WLOG, assume that the $L$ equivalence classes $[\aalpha^1],\ldots, [\aalpha^L]$ are ordered lexicographically by their minimal representatives, $\aalpha^{1*},\ldots, \aalpha^{L*}$. Thus, if $\aalpha^{k*} \geq \aalpha^{\ell*}$, then $k \geq \ell$. Also, let $\xx^\ell = \xxi(Q,[\aalpha^\ell])$ for $\ell = 1,\ldots, L$. Define $T^* = (t^*_{k\ell})$, where $t^*_{k\ell} = t_{\xx^k,[\aalpha^\ell]}$. Then $T^*$ is an $L\times L$ submatrix of $T$, containing the specified rows $\xx^1,\ldots, \xx^L$.  Moreover, $T^*$ is an upper triangular matrix. This is a consequence of the fact that for any $k>\ell$, $\aalpha^{\ell*} \not\geq \aalpha^{k*}$. Thus, there must be some item $j\in\{1,\ldots, J\}$ for which individuals with profiles $\aalpha \in [\aalpha^k]$ possess the necessary attributes, but individuals with profiles $\aalpha \in [\aalpha^\ell]$ do not. Then $ p(X_j=1|[\aalpha^\ell]) = g_j = 0 \Rightarrow t_{\xx^k,[\aalpha^\ell]} = t^*_{k,\ell} = 0$.
 In addition, on the diagonal, $t^*_{\ell,\ell} = \prod_{\{i: x_j^\ell =1\}}(1-s_j) \neq 0$. Thus, $T^*$ is a rank $L$ matrix, as is $T$.
 
 Next suppose that $g\not\equiv 0$. Consider the $T$ matrix as a function of $\cc   = 1- \sb$ and $\gg$. Then the $T$-matrix $T(\cc,\gg)$ can be written as a linear transformation of another $T$-matrix $T(\cc-\gg,\bfzero)$. For any subset of the items $S$ and any constants $b_j$, 
 $$\prod_{j\in S} (b_j-g_j) = \prod_{j\in S} b_j  - \sum_{j\in S} g_j \prod_{k\neq j} b_k + \sum_{k\neq j \in S} g_jg_k \prod_{\ell\neq j,k} b_\ell - \cdots + (-1)^{\#S}\prod_{j\in S} g_j$$
For each entry of the $T$-matrix, the $b_j$ will correspond to either $c_j$ or $g_j$, depending on the value of $\xi_j(Q,[\aalpha])$. Then, $T(\cc-\gg,0) = D(\gg)\cdot T(\cc,\gg)$, where the transformation matrix $D(g)$ is a $2^J\times 2^J$ matrix depending solely on $\gg$. Since the rows of $T$ are ordered lexicographically by $\xx$, $D(g)$ is lower triangular with diagonal $diag(D) \equiv 1$. Thus, $D$ is full-rank and $\text{rank}(T(\cc,\gg)) = \text{rank}(T(\cc-\gg,\bfzero)) = L$.
The model is identifiable, and all other conditions for the consistency of the maximum likelihood estimator are clearly evident.
\end{proof}

\begin{proof}[Proof of Proposition~\ref{prop:marg}]
Since
$$ \sum_{\{\aalpha:\delta_{[\aalpha],k}=1\}}\nu_{\aalpha}
= \sum_{\{[\aalpha]:\delta_{[\aalpha],k}=1\}}\sum_{\aalpha'\in[\aalpha]}\nu_{\aalpha}
= \sum_{\{[\aalpha]:\delta_{[\aalpha],k}=1\}}\nu_{[\aalpha]},$$
$\zeta_k$ can be written in terms of $\nu_{[\aalpha]}$ as
$$\zeta_k = \sum_{\{[\aalpha]:\delta_{[\aalpha],k}=1\}}\nu_{[\aalpha]}.$$
By Theorem~\ref{thm:consistent}, the MLE $\hat\nu_{[\aalpha]}$ is consistent as $N\rightarrow \infty$ under the conditions of the proposition. Thus,
$\hat\zeta_k$
is consistent as $N\rightarrow \infty$.
\end{proof}

\vspace{\fill}\pagebreak

\bibliographystyle{apacite}
\bibliography{IRT2}

\vspace{\fill}\pagebreak
%
%
%
%
%
%
%
\end{document}